\documentclass[12pt, draft, onecolumn, letterpaper]{IEEEtran}

\usepackage{amsfonts, amssymb}
\usepackage{amscd}
\usepackage{amsmath}
\usepackage{comment}
\usepackage{pdflscape}
\usepackage{amsthm}
\usepackage{color}

\newcommand{\Z}{\mathbb{Z}}
\newcommand{\add}{{\mathbb{Z}_2\mathbb{Z}_4}}
\newcommand{\C}{{\cal C}}
\newcommand{\zero}{{\mathbf{0}}}
\newcommand{\cG}{{\cal G}}

\newcommand{\vv}{{\bf  v}}
\newcommand{\vu}{{\bf  u}}
\newcommand{\vw}{{\bf  w}}

\newcommand{\D}{{\cal D}}
\newcommand{\K}{{\cal K}}
\newcommand{\R}{{\cal R}}

\newcommand{\rank}{\operatorname{rank}}

\newtheorem{theorem}{Theorem}
\newtheorem{definition}[theorem]{Definition}
\newtheorem{corollary}[theorem]{Corollary}
\newtheorem{proposition}[theorem]{Proposition}
\newtheorem{lemma}[theorem]{Lemma}
\newtheorem{example}{Example}
\renewcommand{\arraystretch}{0.90}

\title{$\add$-Additive Cyclic Codes: Kernel and Rank}

\author{ Joaquim Borges\thanks{J. Borges is with the Department of Information and Communications
Engineering, Universitat Aut\`{o}noma de Barcelona, 08193-Bellaterra, Spain
(e-mail: joaquim.borges@uab.cat).}  \\
  Steven T. Dougherty\thanks{S. T. Dougherty is with the Department of Mathematics, University of Scranton, Scranton, PA 18510, USA (e-mail:prof.steven.dougherty@gmail.com).} \\
  Cristina Fern\'andez-C\'ordoba\thanks{C. Fern\'andez-C\'ordoba is with the Department of Information and
Communications Engineering, Universitat Aut\`{o}noma de Barcelona,
08193-Bellaterra, Spain (e-mail: cristina.fernandez@uab.cat).}
\\
  Roger Ten-Valls \thanks{R. Ten-Valls is with the Department of Information and Communications
Engineering, Universitat Aut\`{o}noma de Barcelona, 08193-Bellaterra, Spain
(e-mail: roger.ten@uab.cat).}
\thanks{This work has been partially
supported by the Spanish MINECO grants TIN2016-77918-P and MTM2015-69138-REDT,
and by the Catalan AGAUR grant 2014SGR-691.}}

\date{\today}

\begin{document}

\maketitle

\begin{abstract}
A ${\mathbb{Z}}_2{\mathbb{Z}}_4$-additive code ${\cal C}\subseteq{\mathbb{Z}}_2^\alpha\times{\mathbb{Z}}_4^\beta$ is called cyclic if the set of coordinates can be partitioned into two subsets, the set of ${\mathbb{Z}}_2$ coordinates and the set of ${\mathbb{Z}}_4$ coordinates, such that any cyclic shift of the coordinates of both subsets leaves the code invariant. Let $\Phi(\C)$ be the binary Gray map image of $\C$. We study the rank and the dimension of the kernel of a $\Z_2\Z_4 $-additive cyclic code $\C$, that is, the dimensions of the binary linear codes $\langle \Phi(\C) \rangle$ and $\ker(\Phi(\C))$. We give upper and lower bounds for these parameters. It is known that the codes $\langle \Phi(\C) \rangle$ and $\ker(\Phi(\C))$ are binary images of $\Z_2\Z_4$-additive codes $\R(\C)$ and $\K(\C)$, respectively. Moreover, we show that $\R(\C)$ and $\K(\C)$ are also cyclic and we determine the generator polynomials of these codes in terms of the generator polynomials of the code $\C$.
\end{abstract}
\section{Introduction}

Denote by  $\Z_2$ and $\Z_4$ the rings of integers modulo 2 and modulo 4, respectively.  We denote the space of $n$-tuples over these rings as $\Z_2^n$ and $\Z_4^n$. A binary code is any non-empty subset $C$ of $\Z_2^n$, and if that subset is a vector space then we say that it is a linear code.  Any non-empty subset $\C$ of $\Z_4^n$ is a quaternary code and a submodule of $\Z_4^n$ is called a linear code over $\Z_4$.

In 1994, Hammons et al. discovered that some good non-linear binary codes can be seen as the Gray map images of linear codes over $\Z_4$, \cite{sole}. From then on, the study of codes over $\Z_4$ and other finite rings has been developing and the construction of Gray maps has been a topic of study.

In Delsarte's 1973 paper (see \cite{del}),  he defined additive codes as  subgroups of the underlying abelian
group in a translation association scheme. For the binary Hamming scheme, namely when the underlying abelian group is of order $2^{n}$, the only structures for the abelian group are those of the form $\Z_2^\alpha\times \Z_4^\beta$, with $\alpha+2\beta=n$. This means that  the subgroups $\C$ of $\Z_2^\alpha\times \Z_4^\beta$ are the only additive codes in a binary Hamming scheme.
Hence, the study of codes in $\Z_2^\alpha\times \Z_4^\beta$ became important. These codes are called $\add$-additive codes. In recent times, the structure and properties of $\add$-additive codes have been intensely studied (see \cite{AddDual}).

In \cite{Abu}, $\add$-additive cyclic codes were introduced and in \cite{Z2Z4CDual}, the duality was studied. In this paper, we study the rank and the kernel of $\add$-additive cyclic codes, taking into account the known results for general codes over $\Z_4$ and $\add$-additive codes \cite{Z4CyclicRK,Z4RK,Z2Z4RK}.

The paper is organized as follows. In Section \ref{preliminaries}, we recall the necessary concepts and properties on $\add$-additive codes and $\add$-additive cyclic codes. In Section \ref{rankkernel}, we give the main results of the paper about the rank and the kernel of $\add$-additive cyclic codes. We prove that both the binary span and the kernel are binary images of $\add$-additive cyclic codes, we determine the possible values of the dimensions of the corresponding binary images, and we compute the generator polynomials of these codes.

\section{Preliminaries}\label{preliminaries}
\subsection{$\add$-additive codes}

A ${\mathbb{Z}}_2{\mathbb{Z}}_4$-additive code ${\cal C}$ is a subgroup of
${\mathbb{Z}}_2^\alpha\times{\mathbb{Z}}_4^\beta$ (see \cite{AddDual}). Since ${\cal C}$ is a
subgroup of $\mathbb{Z}_2^\alpha\times\mathbb{Z}_4^\beta$, it is also isomorphic
to a commutative structure of the form $\mathbb{Z}_2^\gamma\times\mathbb{Z}_4^\delta$
and it has $|{\cal C}| = 2^{\gamma +2\delta}$ codewords.

Let $X$ (respectively $Y$) be the set of $\mathbb{Z}_2$ (respectively
$\mathbb{Z}_4$) coordinate positions, so $|X| =\alpha$ and $|Y| = \beta$. Unless
otherwise stated, the set $X$ corresponds to the first $\alpha$ coordinates and
$Y$ corresponds to the last $\beta$ coordinates. Let
$\C_X $ be the binary punctured code of ${\cal C}$ formed by deleting the
coordinates outside $X$. Define similarly the quaternary code $\C_Y$.

Let ${\cal C}_b$ be the subcode of ${\cal C}$ generated by all order two codewords and let $\kappa$ be the
dimension of $({\cal C}_b)_X$, which is a binary linear code. For the case $\alpha = 0$, we write $\kappa = 0$. With all these parameters, we say that a code ${\cal C}$ is of type $(\alpha, \beta; \gamma, \delta; \kappa).$

For a vector ${\bf u} \in \Z_2^\alpha\times\Z_4^\beta$ we write
${\bf u}=(u\mid u')$, where
$u=(u_0,\dots,u_{\alpha-1})\in\Z_2^\alpha$ and
$u'=(u'_0,\dots,u'_{\beta-1})\in\Z_4^\beta$.

In \cite{AddDual}, it is shown that  a $\add$-additive code is permutation
equivalent to a
$\add$-additive code with standard generator
matrix of the form:
\begin{equation}\label{eq:StandardForm}
\cG_S= \left ( \begin{array}{cc|ccc}
I_{\kappa} & T_b & 2T_2 & \zero & \zero\\
\zero & \zero & 2T_1 & 2I_{\gamma-\kappa} & \zero\\
\zero & S_b & S_q & R & I_{\delta} \end{array} \right ),
\end{equation} \noindent where $I_k$ is the identity matrix of size $k\times k$; $T_b, S_b$
are matrices over $\Z_2$;  $T_1, T_2, R$ are
matrices over $\Z_4$ with all entries in $\{0,1\}\subset\Z_4$;
and $S_q$ is a matrix over $\Z_4$.

A $\add$-additive code $\C$ is said to be separable if $\C = \C_X \times \C_Y$.  Otherwise the code is said to be non-separable.

Let $u'=(u'_0,\dots, u'_{n-1})$ be an element of $\mathbb{Z}_4^n$ such that $u'_i=\tilde{u}'_i + 2 \hat{u}'_i$, for $i=0, \dots, n-1$ and with $\tilde{u}'_i, \hat{u}'_i\in \{0,1\}.$ As in \cite{sole}, the \emph{Gray map} $\phi$ from $\mathbb{Z}_4^n$ to $\mathbb{Z}_2^{2n}$ is defined by
$$\phi({ u}')= (\hat{u}'_0,\dots, \hat{u}'_{n-1}, \tilde{u}'_0 + \hat{u}'_0,\dots, \tilde{u}'_{n-1}+\hat{u}'_{n-1} ).$$

The \emph{extended Gray map} $\Phi$ is the map from
$\mathbb{Z}_2^\alpha\times\mathbb{Z}_4^\beta$ to $\mathbb{Z}_2^{\alpha+2\beta}$
given by
$$\Phi({u}\mid {u'})= ({ u}\mid \phi({u'})).$$

\subsection{$\add$-additive cyclic codes}

Cyclic codes have been a primary area of study for coding theory, \cite{macwilliams}. Recently, the class of $\add$-additive cyclic codes has been defined in \cite{Abu}.

Let ${\bf u} = (u \mid u' ) \in \Z_2^\alpha \times\Z_4^\beta$, then the cyclic shift $\pi$ is given by
$\pi({\bf u}) = (\pi (u) \mid \pi(u')) $ where $\pi (u) = \pi ( u_0,u_1,\dots,u_{\alpha-1}) = (u_{\alpha-1},u_0,u_1,\dots,u_{\alpha-2})$ and $\pi (u') = (u'_{\beta-1},u'_0,u'_1,\dots,u'_{\beta-2})$ .
We say that a $\add$-additive code $\C$ is cyclic if $\pi(\C) = \C.$

There exists a bijection between $\Z_2^\alpha \times\Z_4^\beta$ and $R_{\alpha,\beta}=\mathbb{Z}_2[x]/(x^\alpha-1)\times\mathbb{Z}_4[x]/(x^\beta-1)$ given by:
\begin{align*}
(u_0, u_1,\dots, u_{\alpha-1}\mid  u'_0,\dots, u'_{\beta-1})&\mapsto\\ (u_0+ u_1x&+\dots +  u_{\alpha-1}x^{\alpha-1}\mid  u'_0+\dots +u'_{\beta-1}x^{\beta-1}).
\end{align*}
Therefore, as usual in the study of cyclic codes, any codeword is identified as a vector or as a polynomial.

From now on, the binary reduction of a polynomial $p(x) \in \mathbb{Z}_4[x]$ will be denoted by $\tilde{p}(x) $.
Let $p (x) \in \Z_4 [x]$ and $ (b(x)\mid a(x)) \in R_{\alpha, \beta}$ and consider the following multiplication $p(x)\star(b(x)\mid a(x))=(\tilde{p}(x)b(x)\mid p(x)a(x)).$ From \cite{Abu}, $R_{\alpha, \beta}$ is a $\Z_4 [x]$-module with respect to this multiplication.

Let ${ u}'(x) =\tilde{{ u}}'(x) + 2\hat{u}'(x) $ be the polynomial representation of ${ u}'\in\mathbb{Z}_4^n$.
Then, the polynomial version of the Gray map is $\phi({ u}'(x) )= (\hat{u}'(x) , \tilde{ u}'(x) +
\hat{u}'(x) ).$ In the following, a polynomial $p (x) \in \Z_2 [x]$ or $\Z_4 [x]$ will be denoted simply by $p $.

Using the polynomial representation, an equivalent definition of $\add$-additive cyclic codes is the following.

\begin{definition}[\cite{Abu}]
	A subset $\C\subseteq R_{\alpha,\beta}$ is called a $\add$-additive cyclic code if $\C$ is a $\Z_4[x]$-submodule of $R_{\alpha,\beta}$.
\end{definition}

From \cite{Abu}, if $\beta$ is odd, we know that if $\C$ is a $\add$-additive cyclic code then it is of the form
\begin{equation} \label{form}
\langle (b\mid { 0}), (\ell \mid fh +2f) \rangle,
\end{equation}
where $fhg = x^\beta -1$ in $\Z_4[x]$, $b$ divides $x^\alpha-1$ in $\Z_2[x]$, and we can assume that $deg(\ell) < deg(b).$ The polynomials satisfying these conditions are said to be in standard form. In this case, we have that $|C| = 2^{\alpha - deg(b)} 4^{deg(g) } 2^{deg(h)}.$ From now on, we assume that $\beta$ is odd.  Then $f$, $g$ and $h$ are pairwise coprime polynomials. Since $h$ and $g$ are
coprime, there exist polynomials $\lambda$ and $\mu$, that will be used later
along the paper, such that
\begin{equation}\label{eq:lambda-mu}
\lambda h+\mu g=1.
\end{equation}

\begin{lemma}[{\cite[Corollary 2]{Z2Z4CDual}}]\label{bdiviXSgcd}
	Let ${\cal C}$ be a ${\mathbb{Z}_2 {\mathbb{Z}_4}}$-additive cyclic code of type $(\alpha, \beta; \gamma , \delta; \kappa)$ with ${\cal C} = \langle (b \mid { 0}), (\ell  \mid  f h  +2f ) \rangle$. Then, $b $ divides $\frac{x^\beta -1}{\tilde{f} } \gcd(b ,\ell )$ and $b $ divides $\tilde{h}  \gcd(b ,\ell \tilde{g} ).$
\end{lemma}
%

We can put the generator matrix (\ref{eq:StandardForm}) in the following form, called the standard form:
\begin{equation} \label{gen}
\left(\begin{array}{ccc|ccc}
I_{\kappa_1} & T& T_{b_1}& {\bf 0}&{\bf 0}&{\bf 0} \\
{\bf 0} & I_{\kappa_2} & T_{b_2} & 2T_2 & {\bf 0} & {\bf 0}  \\
{\bf 0} & {\bf 0} & {\bf 0} & 2T_1 & 2 I_{\gamma - (\kappa_1 + \kappa_2)}& {\bf 0} \\
{\bf 0} & {\bf 0} & S' & S  & R& I_{\delta}
\end{array}\right).
\end{equation}

The next theorem relates the parameters of the type of a
$\add$-additive code to its generator polynomials.

\begin{theorem}[{\cite[Theorem 5 and Proposition 6]{Z2Z4CDual}}]\label{TypeDependingDeg}
	Let $\C$ be a $\add$-additive cyclic code of type $(\alpha, \beta; \gamma , \delta; \kappa)$ with $\C = \langle (b\mid { 0}), (\ell \mid fh +2f) \rangle ,$ where $fhg = x^\beta -1.$
	Then
	\begin{align*}
	\gamma &= \alpha -\deg(b)+\deg(h),\\
	\delta &= \deg(g),\\
	\kappa &= \alpha -\deg(\gcd(\ell \tilde{g}, b))
	\end{align*}
	and
	$$\kappa_1= \alpha -\deg(b),\quad \kappa_2= \deg(b)-\deg(\gcd(b, \ell \tilde{g})),$$
	$$\delta_1= \deg(\gcd(b,\ell \tilde{g})) - \deg(\gcd(b,\ell)) \mbox{ and }
	\delta_2=\deg(g)-\delta_1.$$
\end{theorem}

It is well known that if $\C$ is a $\add$-additive code, then the
$\add$-linear code $C=\Phi(\C)$ is not linear in general. The linearity
of these codes was studied in \cite{Z2Z4RK}. The key to establish this
linearity was the fact that
\begin{equation} \label{equation:linear}
\Phi(\vv+\vw) =
\Phi(\vv) + \Phi(\vw) + \Phi(2\vv*\vw),
\end{equation}
where $*$ denotes the component-wise product. It follows immediately that $\Phi(\C)$ is
linear if and only if $2 {\bf u} * \vv \in \C$, for all ${\bf u}, \vv\in\C$.

It is shown in \cite{Z2Z4Image} that, for a $\add$-additive cyclic code $\C$, $\C_X$
is a binary cyclic code and $\C_Y$ is a linear cyclic code over $\Z_4$.
Moreover, if $\Phi(\C)$ is linear, then $\phi(\C_Y)$ is also linear but the converse
is not true in general. The characterization of linear cyclic codes over $\Z_4$
of odd length whose Gray map images are linear binary codes was given in
\cite{wolfmann}. Let ${p} $ be a divisor of $x^n-1$ in $\mathbb{Z}_2[x]$ with $n$ odd and let $\xi$ be a primitive $n$th root of unity over $\mathbb{Z}_2$. The polynomial $({p}\otimes{p}) $ is defined as the divisor of $x^n-1$ in $\mathbb{Z}_2[x]$ whose roots are the products $\xi^i\xi^j$ such that $\xi^i$ and $\xi^j$ are roots of ${p} .$

\begin{theorem}[{\cite[Theorem 20]{wolfmann}}]\label{Phi(CY)Linear}
	Let ${\cal D}=\langle f h  +2f  \rangle$ be a ${{\mathbb{Z}_4}}$-additive
	cyclic
	code of odd length $n$ and where $f h g  = x^n -1.$ The following properties
	are
	equivalent.
	\begin{enumerate}
		\item $\gcd(\tilde{f} , (\tilde{g}\otimes\tilde{g}) )=1$ in $\mathbb{Z}_2[x]$;
		\item $\phi ({\cal D})$ is a binary linear code of length $2n$.
	\end{enumerate}
\end{theorem}

This result was generalized for $\add$-additive cyclic codes of type
$(\alpha,\beta;\gamma,\delta;\kappa)$ with $\beta$ odd.

\begin{theorem}[\cite{Z2Z4Image}]\label{Phi(C)Linear}
	Let ${\cal C}=\langle (b \mid 0), (\ell \mid f h  +2f ) \rangle$ be a
	$\add$-additive cyclic code of length $\alpha+\beta$, $\beta$ odd, and where $f
	h g  = x^\beta -1.$ The following properties are equivalent.
	\begin{enumerate}
		\item $\gcd(\frac{\tilde{f}b}{\gcd(b, \ell \tilde{g})},
		(\tilde{g}\otimes\tilde{g}) )=1$ in $\mathbb{Z}_2[x]$;
		\item $\Phi ({\cal C})$ is a binary linear code of length $\alpha + 2\beta$.
	\end{enumerate}
\end{theorem}

As a result, it is completely characterized when a $\add$-additive code $\C$
has binary linear image under the Gray map, just considering its generator
polynomials. The next step is study the rank and the dimension of the
kernel for those codes whose image $\Phi(\C)$ is not linear.

\section{Kernel and Rank of $\add$-Additive Cyclic Codes}\label{rankkernel}

For an additive code $\C\subseteq\Z_2^\alpha\times\Z_4^\beta$, the kernel of $\Phi(\C)$ is defined as $\ker(\Phi(\C)) = \{ v\in \Z_2^{\alpha+2\beta}
\ \mid v + \Phi(\C) = \Phi(\C)\}$. Define $\K(\C) = \{ \vv \in
\Z_2^\alpha\times \Z_4^\beta\mid \Phi(\vv) \in \ker(\Phi(\C))\}$.
Let $\rank(\Phi(\C)) = \dim(\langle \Phi(\C) \rangle)$ and $\R(\C) = \{ \vv \mid \ \vv \in \Z_2^\alpha\times \Z_4^\beta, \phi(\vv) \in
\langle \Phi(\C) \rangle \}. $  It is clear that $\K(\C) \subseteq \C \subseteq \R(\C).$

It is known that if $\C$ is a $\add$-additive code, then
$\langle\Phi(\C)\rangle$ and $\ker(\Phi(\C))$ are both $\add$-linear codes
(\cite{Z2Z4RK}).
Therefore, $\R(\C)$ and $\K(\C)$ are both $\add$-additive codes. In the next
sections we will see that if the code $\C$ is a $\add$-additive cyclic code,
then $\R(\C)$ and $\K(\C)$ are also $\add$-additive cyclic. Therefore, the
following proposition will be useful to relate the generator polynomials of
$\C$ to the generator polynomials of $\R(\C)$ and $\K(\C)$.


\begin{proposition}\label{f'_div_f}
	Let $\C_0=\langle  (b\mid 0), (\ell\mid fh +2f )\rangle$ and $\C_1=\langle
	(b'\mid
	0), (\ell'\mid f'h' +2f') \rangle$ be $\add$-additive cyclic codes with
	$\C_0\subseteq \C_1$. Then
	\begin{enumerate}
		\item $f'$ divides $f$;
		\item $\gcd(b',\ell')$ divides $\gcd(b,\ell)$.
	\end{enumerate}
\end{proposition}
\begin{proof}
	Since $\C_0\subseteq \C_1$, we have that $(\C_0)_Y=\langle fh +2f
	\rangle\subseteq (\C_1)_Y=\langle f'h' +2f' \rangle$. Therefore, by
	\cite[Theorem 3]{Z4CyclicRK}, $f'$ divides $f$.
	
	Finally, since $\C_0$ and $\C_1$ are cyclic $\add$-additive codes, then clearly $(\C_0)_X= \langle \gcd(b,\ell)  \rangle$ and  $ (\C_1)_X = \langle
	\gcd(b',\ell')\rangle$. Then, since
	$(\C_0)_X \subseteq (\C_1)_X $, $\gcd(b',\ell')$ divides $\gcd(b,\ell)$.
\end{proof}

In order to study the rank and the kernel of a $\add$-additive code $\C$, it is
necessary to consider the code $\C_b$.

\begin{proposition}[\cite{Z2Z4Image}]\label{2GenOrderTwoCode}
	Let $\C$ be a $\add$-additive cyclic code with $\C = \langle (b\mid{0}), (\ell
	\mid fh +2f) \rangle$. Then
	$\C_b=\langle (b\mid{ 0}), (\bar{\mu}\ell \tilde{g} \mid 2f)\rangle.$
\end{proposition}

Note that if $\C=\C_b$, then (\ref{equation:linear}) is satisfied and
the code $\Phi(\C)$ is linear. In this case, $\delta=0$ and, by
Theorem \ref{TypeDependingDeg}, $g=1$. Therefore $\C_b=\langle (b\mid{ 0}),
(\ell\mid 2f)\rangle.$

\subsection{Kernel of $\add$-Additive Cyclic Codes}

In this section, we will study the kernel of $\add$-additive
cyclic codes. We will prove that, for a $\add$-additive cyclic code $\C$, the code $\K(\C)$ is also cyclic and we will establish some properties of its generator polynomials. We will show that there does not exist a $\add$-additive cyclic code for all the possible values of the dimension of the kernel as for general $\add$-additive codes.

Let ${\cal C}$ be a $\add$-additive code. By (\ref{equation:linear}), we can give the following definition of $\K(\C)$ (see \cite{Z2Z4RK}):

$$
\K(\C)=\{\vv\in\C\mid 2\vv*\vw\in\C, \forall\vw\in\C\}.
$$

\begin{lemma}\label{lemm:KerCY}
	Let ${\cal C}$ be a $\add$-additive code. Then, $\K(\C)_Y\subseteq\K(\C_Y)$.
\end{lemma}	
\begin{proof}
	Let $\C$ be a $\add$-additive cyclic. Let $\vv=(v\mid v')\in\C$. We have that $\vv\in\K(\C)$ if and only if $2\vv*\vw\in\C$, $\forall\vw=(w\mid w')\in\C$.   Since $ 2\vv *\vw=(0\mid 2v'*w')$, we have that if $\vv\in\K(\C)$ then $v'\in\K(\C_Y)$ and the statement follows.
\end{proof}

\begin{example}\label{ex:inequalityKernel}
	Let $\C$ be the $\add$-additive cyclic code in $\frac{\Z_2[x]}{\langle
		x-1\rangle}\times \frac{\Z_4[x]}{\langle x^3-1\rangle}$ 
	generated by $\langle
	(1\mid x+1)\rangle$, where $f=1$ and $h=x-1$. Note that $\C$ is of type
	$(1,3;1,2;1)$ and the generator matrix of $\C$ in standard form,
	(\ref{eq:StandardForm}), is
	$$
	\left(\begin{array}{c|ccc}
	1 & 2 & 0 & 0  \\
	0 & 3 & 1 & 0  \\
	0 & 3 & 0 & 1
	\end{array}\right).
	$$
	
	Since $f=1$, by Theorem \ref{Phi(CY)Linear}, we know that
	$\K(\C_Y)=\C_Y=\langle
	(x+1)\rangle$. We have that the generator matrix of $\K(\C)$ in standard 
	form is
	$$
	\left(\begin{array}{c|ccc}
	1 & 2 & 0 & 0  \\
	0 & 2 & 2 & 0  \\
	0 & 2 & 0 & 2
	\end{array}\right)
	$$
	and therefore $\K(\C)=\langle (1\mid 2)\rangle$. Hence,
	$\K(\C)_Y\varsubsetneq \K(\C_Y)$.
\end{example}

The following theorems determine an upper and a lower bound for the kernel of a
$\add$-additive code and that there exists a $\add$-additive code of type
($\alpha,\beta, \gamma,
\delta, \kappa$) for all possible values of the kernel.

\begin{theorem}[\cite{Z2Z4RK}]\label{minandmax}
	Let $\C$ be a $\add$-additive code with parameters ($\alpha,\beta, \gamma,
	\delta, \kappa$).  Then $\gamma + \delta
	\leq \dim (\ker(\Phi(\C)) \leq \gamma + 2 \delta.$
\end{theorem}

\begin{theorem}[\cite{Z2Z4RK}] \label{all-kernels} Let $\alpha,\beta,\gamma,
	\delta,\kappa$ be integers satisfying
	\begin{equation} \label{bounds-code}
	\begin{tabular}{c}
	$\alpha, \beta, \gamma, \delta, \kappa \geq 0$,  $\quad \alpha+\beta >0$, \\
	$0 < \delta+\gamma \leq \beta + \kappa \quad $ \textrm{and} $\quad \kappa
	\le \min(\alpha,\gamma)$.
	\end{tabular}
	\end{equation} Then, there exists a $\add$-linear code $C$ of
	type $(\alpha,\beta;\gamma, \delta;\kappa)$ with 
	$\dim(\ker(C))=\gamma+2\delta-\bar k$
	if and only if
	$$\left\{
	\begin{tabular}{lll}
	$\bar k=0$, &  if & $s=0$,  \\
	$\bar k \in \{0\} \cup \{2, \ldots, \delta \} \mbox{ and $\bar k$ even},$ &
	if & $s=1$, \\
	$\bar k \in \{0\} \cup \{2,  \ldots , \delta \},  $ & if &$ s\geq 2$,
	\end{tabular} \right.$$ where $s=\beta-(\gamma-\kappa)-\delta$.
\end{theorem}

We will see that not all possible values for the kernel of a $\add$-additive 
code
$\C$ are possible if $\C$ is cyclic. First, we will determine some properties of
the kernel of a $\add$-additive cyclic code.

\begin{proposition}
	Let $\C=\langle  (b\mid 0), (\ell\mid fh +2f )\rangle$ be a $\add$-additive 
	cyclic code. Then
	$$\alpha -\deg(b)+\deg(h) + \deg(g)
	\leq |\K (\C)| \leq \alpha -\deg(b)+\deg(h) + 2 \deg(g).$$
\end{proposition}
\begin{proof}
	Straightforward from Theorems \ref{minandmax} and \ref{TypeDependingDeg}.
\end{proof}

Note that the upper bound is sharp when the code has binary linear image, i.e.,
$\C=\K(\C)$.  Moreover, the lower bound is tight when $\K(\C)$ only has the 
all-zero
vector and all order two codewords; that is, $\C=\C_b$.

\begin{proposition}\label{K(C)Y_in_K(CY)}
	Let $\C$ be a $\add$-additive cyclic code, then $\K(\C) \subseteq
	\C_X\times \K(\C_Y).$
\end{proposition}
\begin{proof}
	Let $\vv=(v,v') \in \K(\C)$.  For all $\vw=(w,w') \in \C$, $2 \vv*\vw 
	\in
	\C$ since $\vv $ is in the kernel.  Then $v\in\C_X$ and
	$2 v' * w' \in \C_Y$, for all $w' \in \C_Y$ which gives $v' \in 
	\K(\C_Y)$ and
	$\vv\in \C_X\times \K(\C_Y)$ and the result follows.
\end{proof}

\begin{proposition} If $\C$ is a separable $\add$-additive code then $\K(\C) =
	\C_X \times \K(\C_Y).$
\end{proposition}
\begin{proof}
	Let $\textbf{v}\in\C_X \times \K(\C_Y)$, then $2\textbf{v}*\textbf{w}=({\bf 0 
	}\mid 2v'*w')$ for all $\textbf{w}\in\C$. Since $v'\in\K(\C_Y)$, then 
	$2v'*w'\in\C_Y$. Moreover, since $\C$ is a separable $\add$-additive cyclic 
	code,  $2\textbf{v}*\textbf{w}=({\bf 0 }\mid 2v'*w')\in\C$. Therefore 
	$\textbf{v}\in\K(\C)$ and $\C_X \times \K(\C_Y)\subseteq \K(\C)$.
	
	Now, by Proposition~\ref{K(C)Y_in_K(CY)}, we have $\K(\C)\subseteq \C_X \times
	\K(\C_Y)$.
	$\textbf{w}\in\C$. Since $\C$ is a separable $\add$-additive cyclic code, 
	$2v'*w'\in\C_Y$ for all $w'\in\C_Y$ and so $v'\in\K(\C_Y)$. Therefore 
\end{proof}

The following example shows that if the $\add$-additive code $\C$ is
non-separable, then the kernel is not necessarily $\C_X \times \K(\C_Y).$

\begin{example}
	Let $\C=\langle (1\mid x+1)\rangle$, the $\add$-additive code of 
	Example~\ref{ex:inequalityKernel}, with $\C_X=\langle 1 \rangle$ and 
	$\C_Y=\langle x+1\rangle$. We have seen that $\K(\C_X)=\C_X$, $\K(\C_Y)=\C_Y$ 
	and $\K(\C)=\langle (1\mid 2) \rangle$. Therefore $\K(\C)\subsetneq \C_X \times
	\K(\C_Y)$.
\end{example}

Therefore, if the code is non-separable, the equality is not satisfied in
general.

From Proposition \ref{K(C)Y_in_K(CY)}, we obtain that $\dim(\ker(\Phi(\C))) \leq
\kappa + \dim(\ker(\phi(\C_Y))).$ However, we can give a bound that is more
accurate.

\begin{proposition}\label{K(C)<K(C_Y)+k1}
	Let $\C$ be a $\add$-additive code, then $\dim(\ker(\Phi(\C))) \leq
	\kappa_1+ \dim(\ker(\phi(\C_Y))).$
\end{proposition}
\begin{proof}
	Define $\C_0 = \{ \vv=(v\mid v') \in \C \mid v' ={\bf 0 } \}$.
	We have that $\C_0 \subseteq \K(\C)$, and $\dim(\ker(\C_0))=\kappa_1$.
	
	Let $\vv=(v\mid v')\in\K(\C)$. If $\vv'={\bf 0}$, then $\vv\in\C_0$. Otherwise, 
	$v'\in\K(\C)_Y\subseteq\K(\C_Y)$ by Lemma~\ref{lemm:KerCY} and, therefore, 
	$\dim(\ker(\Phi(\C))) \leq \kappa_1+\dim(\ker(\phi(\C_Y))).$
\end{proof}

From the generator matrix $G$ of $\C$ given in (\ref{gen}), we have that the 
code $\C_Y$ has a generator matrix of the form
\begin{equation} \label{genC_Y}
\left(\begin{array}{ccc}
2T_2 & {\bf 0} & {\bf 0}  \\
2T_1 & 2 I_{\gamma - (\kappa_1 + \kappa_2)}& {\bf 0} \\
S  & R& I_{\delta}
\end{array}\right).
\end{equation}

By  \cite[Proposition 1]{Z2Z4CDual}, we know that the code $\C_Y$ has type 
$4^\delta 2^{\gamma - \kappa_1}$. The minimum value for the dimension of 
$\ker(\phi(\C_Y))$ is $\delta + \gamma - \kappa_1$.

\begin{theorem} Let $\C$ be a $\add$-additive code. If $\K(\C_Y)$ is a minimum 
	then $\K(\C)$ is a minimum.
\end{theorem}
\begin{proof}
	If $\K(\C_Y)$ is a minimum, by Proposition \ref{K(C)<K(C_Y)+k1}, then 
	$\dim(\ker(\Phi(\C))) \leq \kappa_1+ \delta + \gamma - \kappa_1 = \gamma + 
	\delta$.
\end{proof}

In the previous statements, we compute an upper bound of the kernel of a
$\add$-additive code $\C$ by considering the kernel of a code over $\Z_4$.
Now we shall give the exact value of the kernel of a $\add$-additive code $\C$
in terms of $\C_X$ and the kernel of a linear code over $\Z_4$, $\C'$. As we
have seen, in the case of a separable code $\C$, the code $\C'$ is exactly 
$\C_Y$
and the value $\dim(\ker(\Phi(\C)))$ is 
$\dim(\ker(\Phi(\C_X)))+\dim(\ker(\Phi(\C_Y)))$, where $\C_Y$ is a
cyclic code over $\Z_4$. If the code $\C$ is not separable, then  $\C'$ is not
necessarily cyclic.

Let $\C$ be a $\add$-additive code with generator matrix in the form of
(\ref{gen}) and let $\C'$ be the subcode generated by
\begin{equation} \label{genSubcode}
\left(\begin{array}{ccc|ccc}
{\bf 0} & {\bf 0} & {\bf 0} & 2T_1 & 2 I_{\gamma - (\kappa_1 + \kappa_2)}& {\bf
	0} \\
{\bf 0} & {\bf 0} & S' & S  & R& I_{\delta}
\end{array}\right).
\end{equation}

\begin{theorem}
	Let $\C$ be a $\add$-additive with generator matrix in the form of
	(\ref{gen}) and let $\C'$ be the subcode generated by the matrix in 
	(\ref{genSubcode}). Then
	$\dim(\ker(\Phi(\C)) = \kappa_1 + \kappa_2 + \dim (\ker(\phi(\C'_Y)))$.
\end{theorem}

\begin{proof}
	Let $\C$ be a $\add$-additive cyclic code with generator matrix $G$ in the
	form
	of (\ref{gen}). Let $\{\vu_i=(u_i\mid u'_i)\}_{i=1}^{\gamma}$ be the
	first $\gamma$ rows and  $\{\vv_j=(v_j\mid v'_j)\}_{j=1}^{\delta}$ the last
	$\delta$ rows of $G$. Define the codes
	$\bar{\C}=\langle\{\vu_i\}_{i=1}^{\kappa_1+\kappa_2}\rangle,$
	$\C'=\langle\{\vu_i\}_{i=\kappa_1+\kappa_2+1}^{\gamma}
	,\{\vv_j\}_{j=1}^{\delta}\rangle$.
	
	By \cite{Z2Z4RK}, $\vv\in\K(\C)$ if and only if $2\vv*\vw\in\C$ for all
	$\vw\in\C$. We have that $\vv\in\bar{C}$ is of order $2$ and hence $2\vv*\vw=0$, 
	$\forall \vw\in\C$. Then, $\bar{\C}\subseteq \K(\C)$ and
	$\dim(\ker(\bar{\C}))=\kappa_1+\kappa_2$. Let $\vv=(v\mid v')\in\C'$. Since
	$2\vv*\vw=\zero$ for all $\vw\in\bar{\C}$, we have that $\vv\in\K(\C)$ if and
	only if $2\vv*\vw\in\C'$ for all $\vw\in\C'$; that is, $\vv\in\K(\C')$. Finally,
	$2\vv*\vw=(\zero\mid 2v'*w')\in\C'$ if and only if $2v'*w'\in(\C')_Y$,
	and hence $\dim(\ker(\Phi(\C')))=\dim(\ker(\Phi((\C')_Y)))$. Therefore, 
	$\dim(\ker(\Phi(\C))) =
	\kappa_1 + \kappa_2 + \dim (\ker(\phi(\C'_Y)))$.
\end{proof}


Now we will establish the kernel of a $\add$-additive cyclic codes taking into 
account
its generator polynomials. In the following theorem we shall
prove that if $\C$ is $\add$-additive cyclic code, then $\K(\C)$ is also
cyclic. This result is a generalization of the case when $\C$ is a linear
cyclic code over $\Z_4$ that is given in \cite{Z4CyclicRK}.

\begin{theorem}\label{K(C)_cyclic}
	Let $\C$ be a $\add$-additive cyclic code.  Then $\K(\C) $ is a $\add$-additive
	cyclic code.
\end{theorem}

\begin{proof}
	We know that $\K(\C)$ is a $\add$-additive code so we just have to show that if
	$\vu = (u \ | \ u') \in \K(\C)$ then $\pi(\vu) \in \K(\C)$. That is, we want
	to show that $2\pi(\vu)*\vw\in\C$, for all $\vw\in\C$.
	
	Let $\vu \in \K(\C), \vw\in\C$.
	Then  $2 \pi(\vu)*\vw = \pi (2 \vu * \pi^{-1} (\vw))$.
	We have that $\vu \in \K(\C)$ and $\pi^{-1}(\vw)\in
	\C$, therefore $2 \vu * \pi^{-1} (\vw) \in \C$ by
	(\ref{equation:linear}).
	Since
	the code $\C$ is cyclic, $\pi(2 \vu * \pi^{-1} (\vw)) \in \C$, which gives
	that $2 \pi(\vu)*\vw\in\C$, and  $\pi(\vu) \in \K(\C)$.
\end{proof}

\begin{corollary}Let ${\cal C}=\langle (b\mid 0), (\ell\mid fh +2f) \rangle$ be
	a $\add$-additive cyclic code, where $fhg = x^\beta -1.$ Then, $\K(\C)=\langle
	(b_k\mid 0), (\ell_k\mid f_kh_k +2f_k) \rangle$, where $f_kh_kg_k = x^\beta -1$
	and
	\begin{enumerate}
		\item $f$ divides $f_k$;
		\item $\gcd(b,\ell)$ divides $\gcd(b_k,\ell_k)$.
	\end{enumerate}
\end{corollary}

\begin{proof}
	By Theorem \ref{K(C)_cyclic}, $\K(\C)$ is cyclic and therefore $\K(\C)=\langle
	(b_k\mid 0), (\ell_k\mid f_kh_k +2f_k) \rangle$, where $f_kh_kg_k = x^\beta
	-1$. Since $\K(\C)\subseteq \C$, the result follows from
	Proposition~\ref{f'_div_f}.
\end{proof}

Let ${\cal C}=\langle (b\mid 0), (\ell\mid fh +2f) \rangle$ be a $\add$-additive
cyclic code, where $fhg = x^\beta -1.$ In \cite{Z2Z4Image} it is proved that 

\begin{equation}\label{eq:3gen}
\C=\langle (b\mid 0), (\ell'\mid fh),(\tilde{\mu}\ell\tilde{g}
\mid 2f)\rangle,
\end{equation}
where $\ell'=\ell -\tilde{\mu}\ell\tilde{g}$.

\begin{lemma}\label{lemm:lk}
	Let ${\cal C}=\langle (b\mid 0), (\ell\mid fh +2f) \rangle$ be a $\add$-additive
	cyclic code, where $fhg = x^\beta -1.$ Let $\langle (b\mid 0),
	(\ell_k\mid fhk +2f) \rangle\subset\C$, for $k|g$. Then ${\ell}_k=\tilde{k}\ell 
	+
	(1-\tilde{k})\tilde{\mu}\ell\tilde{g}
	\pmod{b}$, for $\mu$ in (\ref{eq:lambda-mu}).
\end{lemma}
\begin{proof}
	Let ${\cal C}=\langle (b\mid 0), (\ell'\mid fh),(\tilde{\mu}\ell\tilde{g}
	\mid 2f)\rangle$, where $\ell'=\ell -\tilde{\mu}\ell\tilde{g}$, as in 
	(\ref{eq:3gen}). Since $(\ell_k\mid fhk +2f)\in\C$ and $(\ell_k\mid fhk
	+2f)=c_1(b\mid 0)+c_2(\ell'\mid fh)+c_3(\tilde{\mu}\ell\tilde{g}\mid 2f)$, we
	obtain $c_2=k$, $c_3=1$ and
	${\ell}_k=\tilde{k}\ell' + \tilde{\mu}\ell\tilde{g}\pmod{b}=\tilde{k}\ell +
	(1-\tilde{k})\tilde{\mu}\ell\tilde{g} \pmod{b}$.
\end{proof}

\begin{theorem}\label{Theo:GeneratorsKernel}
	Let ${\cal C}=\langle (b\mid 0), (\ell\mid fh +2f) \rangle$ be a $\add$-additive
	cyclic code, where $fhg = x^\beta -1.$ Then,
	$\K(\C)=\langle (b\mid 0), ({\ell}_k\mid fhk +2f) \rangle$, where $k$ divides
	$g$ and ${\ell}_k=\tilde{k}\ell + (1-\tilde{k})\tilde{\mu}\ell\tilde{g}
	\pmod{b}$, for $\mu$ in (\ref{eq:lambda-mu}).
\end{theorem}

\begin{proof}
	By Theorem \ref{K(C)_cyclic}, $\K(\C)$ is cyclic and then $\K(\C)=\langle
	({b}_k\mid 0), ({\ell}_k\mid {f}_k{h}_k +2{f}_k) \rangle$. Clearly,
	${b}_k=b$. Since $\C_b\subseteq \K(\C)\subseteq\C$, by Proposition
	\ref{2GenOrderTwoCode} and Proposition \ref{f'_div_f}, we conclude that
	${f}_k=f$. Since $\K(\C)_Y=\langle f{h}_k+2f\rangle\subseteq\C_Y$, with
	an argument analogous to that of \cite[Theorem 9]{Z4CyclicRK} we obtain that 
	${h}_k=hk$
	with $k$ a divisor of $g$.

	Let $\ell'=\ell -\tilde{\mu}\ell\tilde{g}$. By (\ref{eq:3gen}), $(\ell'\mid fh),
	(\tilde{\mu}\ell\tilde{g}\mid 2f)\in\C$. Therefore,
	${\ell}_k=\tilde{k}\ell' + \tilde{\mu}\ell\tilde{g}\pmod{b}=\tilde{k}\ell +
	(1-\tilde{k})\tilde{\mu}\ell\tilde{g} \pmod{b}$.
\end{proof}

Theorem~\ref{all-kernels} shows that there exists a $\add$-additive code for all
possible values of the kernel for a given type $(\alpha, \beta; \gamma,\delta;
\kappa)$. Considering the last theorem, the next example illustrates that this
result is not true for $\add$-additive cyclic codes; i.e., for a given type
$(\alpha, \beta; \gamma,\delta; \kappa)$ there does not always exist a
$\add$-additive cyclic code for all possible values of the kernel.
Furthermore, it shows that there does not always exist a $\add$-additive cyclic 
code
for a given type $(\alpha, \beta; \gamma,\delta; \kappa)$.

\begin{example}\label{ex:Kernel}
	By Theorem~\ref{all-kernels}, there exists a
	$\add$-additive code $\C$ of type $(2,7; $ $2,3;$ $\kappa)$ with
	$\dim(\ker(\Phi(\C)))=k_d$, for all $k_d\in\{5,6,8\}$. We will see that there 
	does
	not exist any cyclic $\add$-additive code of type $(2,7; 2,3;
	\kappa)$, with dimension of the kernel in $\{6,8\}$.

	Let $\alpha=2$ and $\beta=7$. We have that $x^7-1=(x-1)(x^3 + 2x^2 + x + 3)(x^3 
	+ 3x^2 + 2x + 3)$. Let ${\cal C}=\langle (b\mid 0), (\ell\mid fh +2f) \rangle$ 
	be a
	$\add$-additive cyclic code of type $(2,7; 2,3; \kappa)$, where $fhg = x^7 -1$.
	
	By Theorem~\ref{TypeDependingDeg}, $\deg(g)=3$ and $\deg(b)=\deg(h)\leq 2$. Let
	$\{p_3,q_3\}=\{(x^3 + 2x^2 + x + 3),(x^3 + 3x^2 + 2x + 3)\}$. Assume without
	loss of generality that $g=p_3$ and, since $\deg(h)\leq 2$, we have that $q_3$
	divides $f$. It is easy to see that
	$\gcd(q_3,(\tilde{p_3}\otimes\tilde{p_3}))\not=1$ and therefore
	$\gcd(\frac{\tilde{f}b}{\gcd(b, \ell \tilde{g})}, (\tilde{g}\otimes\tilde{g})
	)\neq 1$. Hence, by Theorem~\ref{Phi(C)Linear}, there does not exist a
	$\add$-additive code of type $(2,7; 2,3; \kappa)$ with linear Gray image. Thus,
	$\dim(\ker(\Phi(\C)))\neq 8$.
	
	By Theorem~\ref{Theo:GeneratorsKernel}, $\K(\C)=\langle (b\mid 0),
	({\ell}_k\mid fhk +2f) \rangle$ where $k$ divides $g$. By the previous argument,
	$k\neq 1$ and then we have that $k=g=p_3$ and $\K(\C)=\langle (b\mid 0),
	({\ell}_k\mid 2f) \rangle$. Therefore $\K(\C)$ does not contain codewords of
	order $4$, thus $\dim(\ker(\Phi(\C)))=\gamma + \delta= 5$.
	
	Finally, we will give $\add$-additive cyclic codes of type $(2,7; 2,3;
	\kappa)$, for different values of $\kappa$. Recall that $\kappa\leq
	\min\{\alpha, \gamma\}=2$ and $\kappa= \alpha - \deg(\gcd(b,\ell \tilde{g}))$,
	then
	\begin{itemize}
		\item $\kappa=2$: $\C=\langle (x - 1\mid 0), (1\mid (x^3 + 2x^2 + x + 3)(x-1)
		+2(x^3 + 2x^2 + x + 3))\rangle$, or $\C=\langle (x - 1\mid 0), (1\mid (x^3 +
		3x^2 + 2x + 3)(x -1)
		+ 2(x^3 + 3x^2 + 2x + 3))\rangle$.
		\item $\kappa=1$: $\C=\langle (x - 1\mid 0), (0\mid (x^3 + 3x^2 + 2x + 3)(x -1)
		+ 2(x^3 + 3x^2 + 2x + 3))\rangle$, or $\C=\langle (x - 1\mid 0), (0\mid (x^3 +
		2x^2 + x + 3)(x-1)
		+2(x^3 + 2x^2 + x + 3))\rangle$.
		\item $\kappa=0$: In this case,  $\deg(\gcd(b,\ell \tilde{g}))=2$ and
		therefore, $\gcd(b,\ell \tilde{g})=x^2-1$. Note that $\tilde{p}_3$ and
		$\tilde{q}_3$ are not divisors of $x^2-1$ over $\Z_2$, thus there does not exist
		$\ell$ with $\deg(\ell)<2$ such that $\ell \tilde{g}=x^2-1$. There does not
		exist a $\add$-additive cyclic code of type $(2,7; 2,3; 0)$.
	\end{itemize}
\end{example}

The statement in Theorem \ref{Theo:GeneratorsKernel} is also true for any
maximal $\add$-additive cyclic subcode of a $\add$-additive cyclic code $\C$
whose Gray image is a linear subcode of $\Phi(C)$.

\begin{corollary}\label{Coro:GeneratorsMaximal}
	Let ${\cal C}=\langle (b\mid 0), (\ell\mid fh +2f) \rangle$ be a
	cyclic code, where $fhg = x^\beta -1.$ Then, if $\C_1$ is a maximal
	$\add$-additive cyclic subcode with $\Phi(\C_1)$ linear, then
	$\C_1=\langle (b\mid 0), ({\ell}_k\mid fhk +2f) \rangle$, where $k$ divides
	$g$ and ${\ell}_k=\tilde{k}\ell+(1 -\tilde{k})\tilde{\mu}\ell\tilde{g}
	\pmod{b}$, for $\mu$ in (\ref{eq:lambda-mu}).
\end{corollary}


The kernel of a binary code is the intersection of all its maximal linear 
subspaces.
Therefore, if $\C_1$, $\C_2,\dots \C_s$ are all the maximal subcodes of a
$\add$-additive code $\C$ such that $\phi(\C_i)$ is a linear subcode of
$\phi(\C)$, for $1\leq i \leq s$, then
\begin{equation}\label{eq:kernel}
\K(\C)=\bigcap_{i=1}^{s}\C_i.
\end{equation}

In \cite{Z4CyclicRK} it is proved that if $\C_1=\langle fh_1+2f \rangle$
and $\C_2=\langle fh_2+2f \rangle$ are quaternary cyclic codes of odd
length $n$, then $\C_1\cap\C_2=\langle f\,{\rm lcm}(h_1,h_2)+2f
\rangle$. We will give a similar result for $\add$-additive cyclic codes.

\begin{proposition}\label{lemma:intersection}
	Let $\C=\langle (b\mid 0),(l \mid fh+2f)\rangle$ be a $\add$-additive cyclic
	code. Let $\C_1=\langle (b\mid 0),(l_{k_1} \mid fhk_1+2f) \rangle$ and
	$\C_2=\langle (b \mid 0),(l_{k_2} \mid fhk_2+2f) \rangle$ be $\add$-additive
	maximal subcodes of $\C$ whose images under the Gray map are
	linear subcodes of $\phi(\C)$. Then
	$$
	\C_1\cap\C_2=\langle (b\mid 0),(l_{k'}\mid fhk'+2f) \rangle,
	$$
	where $k'=\,{\rm lcm}(k_1,k_2)$ and ${\ell}_{k'}=\tilde{k'}\ell+(1 
	-\tilde{k'})\tilde{\mu}\ell\tilde{g}
	\pmod{b}$, for $\mu$ in (\ref{eq:lambda-mu}).
\end{proposition}

\begin{proof}
	Let $\C=\langle (b\mid 0),(l \mid fh+2f)\rangle$ be a $\add$-additive cyclic
	code and let $\C_i=\langle (b \mid 0),(l_{k_i} \mid fhk_i+2f) \rangle$ be a
	$\add$-additive
	maximal subcode of $\C$ whose image under the Gray map is a
	linear subcode of $\phi(\C)$.
	
	We first consider $(\C_1)_Y\cap(\C_2)_Y$.
	By \cite{Z4CyclicRK}, $(\C_1)_Y\cap(\C_2)_Y=\langle f\,{\rm lcm}(hk_1,hk_2)+2f
	\rangle=\langle fhk'+2f\rangle$, where $k'={\rm lcm}(k_1,k_2)$. Since
	$\langle(b|0)\rangle\in\C_1\cap\C_2$, then $\C_1\cap\C_2= \langle (b\mid
	0),(l_{k'}\mid fhk'+2f) \rangle$, where ${\ell}_{k'}=\tilde{k'}\ell+(1 
	-\tilde{k'})\tilde{\mu}\ell\tilde{g}
	\pmod{b}$, for $\mu$ in (\ref{eq:lambda-mu}) by Lemma~\ref{lemm:lk}.
\end{proof}

\begin{lemma}\label{inside}
	Let $\C$ be a $\add$-additive cyclic code and let $\D$ be a maximal cyclic 
	subcode with linear binary image. Then, $\K(\C)\subseteq \D$.
\end{lemma}
\begin{proof}
	If $\K(\C)\not\subseteq \D$, then consider the $\add$-additive code $\D'$ 
	generated by $\K(\C)\cup \D\cup \{2\vu * \vv\mid \vu,\vv\in \K(\C)\cup \D\}$. 
	Since the binary image of $\K(\C)\cup \D$ is cyclic, $\D'$ is a cyclic subcode 
	of $\C$. Moreover, since $2\vu * \vv \in \D'$, for all $\vu,\vv\in \D'$, we have 
	that $\D'$ has linear binary image, leading to a contradiction since we are 
	assuming that $\D$ is maximal.
\end{proof}

\begin{theorem}\label{ker:k}
	Let ${\cal C}=\langle (b\mid 0), (\ell\mid fh +2f) \rangle$ be a $\add$-additive
	cyclic code, where $fhg = x^\beta -1$. Assume that $k_1,\dots,k_s$ are the 
	divisors
	of $g$ of minimum degree such that
	$$\gcd\left(\frac{\tilde{f}b}{\gcd(b, \ell \frac{\tilde{g}}{\tilde{k_i}})},
	\left(\frac{\tilde{g}}{\tilde{k_i}}\otimes\frac{\tilde{g}}{\tilde{k_i}}\right)
	\right)=1,$$
	for $i=1,\dots,s$. Then,
	$$
	\K(\C)=\langle (b\mid 0), ({\ell}_k'\mid fhk' +2f) \rangle,
	$$
	where  $k'=\,{\rm lcm}(k_1,\dots,k_s)$ and
	${\ell}_{k'}=\tilde{k'}\ell+(1 -\tilde{k'})\tilde{\mu}\ell\tilde{g}
	\pmod{b}$, for $\mu$ in (\ref{eq:lambda-mu}).
\end{theorem}
\begin{proof}
	Assume that $k_1,\dots,k_s$ are the divisors
	of $g$ of minimum degree such that
	$$\gcd\left(\frac{\tilde{f}b}{\gcd(b, \ell \frac{\tilde{g}}{\tilde{k_i}})},
	\left(\frac{\tilde{g}}{\tilde{k_i}}\otimes\frac{\tilde{g}}{\tilde{k_i}}\right)
	\right)=1,$$
	for $i=1,\dots,s$. Let $\D_i$ be a cyclic subcode of $\C$, where $\D_i=\langle
	(b\mid 0), ({\ell}_{k_i}\mid fhk_i +2f) \rangle$, for some
	$\ell_{k_i}$. Note that $\Phi(\D_i)$ is linear by
	Theorem~\ref{Phi(C)Linear}. Since $k_i$ is a polynomial of minimum degree
	dividing $g$, then $\D_i$ is a maximal cyclic subcode of $\C$ with linear binary
	image. Then each $\D_i$ extends to $\C_i$ which is a maximal subcode of $\C$,
	not necessarily cyclic, with linear binary image. Note that every maximal code
	with linear image must contain a cyclic code
	with linear image, e.g., every maximal code contains $\K(\C)$ that is cyclic
	with linear image. By Lemma \ref{inside}, we know $\K(\C)\subseteq \D_i$
	and, therefore $\K(\C)\subseteq \cap_i \D_i$. But $\cap_i\C_i= \K(\C)\subseteq
	\cap_i \D_i\subseteq \cap_i \C_i$, so $\K(\C)=\cap_i\D_i$. By Corollary
	\ref{Coro:GeneratorsMaximal} and Proposition \ref{lemma:intersection}, the
	result
	follows. \end{proof}

\begin{example}
	Let $x^7-1=(x-1)p_3q_3$ over $\Z_4$. Let $\C=\langle (1\mid 0),(0\mid  f) 
	\rangle$ of type $(1, 7; 1, 6; 1 )$ with $f=(x-1)$, $h=1$ and $g=p_3q_3$.
	
	Note that $\gcd(\frac{\tilde{f}b}{\gcd(b, \ell \tilde{g})}, 
	(\tilde{g}\otimes\tilde{g}))=x-1\neq 1$. We have that all maximal cyclic 
	subcodes of $\C$ with binary linear image are $\C_1=\langle (1\mid 0),(0\mid 
	fp_3 + 2f) \rangle,$ and $ \C_2=\langle (1\mid 0),(0\mid fq_3 + 2f) \rangle.$
	Clearly, $k'=\,{\rm lcm}(p_3, q_3)=p_3q_3$ and then $\K(\C)=\langle (1\mid 
	0),(0\mid fp_3q_3 + 2f) \rangle=\langle (1\mid 0),(0\mid 2f) \rangle$.
\end{example}

\subsection{Rank of $\add$-Additive Cyclic Codes}

In this section, we will study the rank of a $\add$-additive
cyclic code $\C$. We will prove that $\R(\C)$ is also cyclic and we will
establish some properties of its generator polynomials. However, we will show
that there does not exist a $\add$-additive cyclic code for all possible values
of the rank, in contrast to what is exhibited in the following results for a
$\add$-additive code.

\begin{proposition}[\cite{Z2Z4RK}] \label{bounds-rank}
	Let $C$ be a $\add$-linear code of binary length $n=\alpha+2\beta$
	and type $(\alpha,\beta;\gamma, \delta;\kappa)$.  Then,
	$ \rank(\Phi(C))\in \{ \gamma+2\delta,\ldots, \min (\beta+\delta+\kappa, 
	\; \gamma
	+2\delta + \binom{\delta}{2} ) \}.$
\end{proposition}

\begin{theorem}[\cite{Z2Z4RK}] \label{all-ranks} Let $\alpha,\beta,\gamma,
	\delta,\kappa$ be integers satisfying
	(\ref{bounds-code}). Then, there exists a $\add$-linear code $C$ of
	type $(\alpha,\beta;\gamma, \delta;\kappa)$ with $\rank(\Phi(C))=r$ if 
	and only if
	$$r\in \{ \gamma+2\delta, \ldots, \min (\beta+\delta+\kappa, \; \gamma
	+2\delta + {\delta \choose 2} )\}.$$
\end{theorem}

\begin{proposition}[\cite{Z2Z4RK}]\label{lemm:RankSet}
	Let $\C$ be a $\add$-additive code of type
	$(\alpha,\beta;\gamma,\delta;\kappa)$ and let $C=\Phi(\C)$.
	Let $\cG$ be a generator matrix of $\C$ as in
	(\ref{eq:StandardForm}) and let $\{\vu_i\}_{i=1}^{\gamma}$ be the rows of order
	two and
	$\{\vv_j\}_{j=1}^{\delta}$ the rows of
	order four in $\cG$. Then, $\langle C \rangle$ is generated by
	$\{\Phi(\vu_i)\}_{i=1}^{\gamma}$,
	$\{\Phi(\vv_j),\Phi(2\vv_j)\}_{j=1}^{\delta}$ and $\{\Phi(2\vv_j *
	\vv_k)\}_{1\leq j<k\leq \delta}$. Moreover, $\langle C \rangle$ is
	both binary linear and $\add$-linear.
\end{proposition}

\begin{corollary}[\cite{Z4CyclicRK}]\label{coro:generatorRank}
	Let $\C$ be a $\add$-additive code of type
	$(\alpha,\beta;\gamma,\delta;\kappa)$ and let $C=\Phi(\C)$.
	Let $\cG$ be a generator matrix of $\C$ as in
	(\ref{eq:StandardForm}) and let $\{\vu_i\}_{i=1}^{\gamma}$ be the rows 
	of order
	two and
	$\{\vv_j\}_{j=1}^{\delta}$ the rows of
	order four in $\cG$. Then,  
	$$\R(\C)=\C\cup\langle \{2\vv_j *\vv_k\}_{1\leq j<k\leq 
		\delta}\rangle.$$
\end{corollary}

\begin{lemma}\label{lemm:RankCY}
	Let ${\cal C}$ be a $\add$-additive code. Then, $\R(\C)_Y=\R(\C_Y)$.
\end{lemma}	
\begin{proof}
	Let $\{\vu_i=(u_i\mid u'_i)\}_{i=1}^{\gamma}$ be the
	first $\gamma$ rows and  $\{\vv_j=(v_j\mid v'_j)\}_{j=1}^{\delta}$ the 
	last
	$\delta$ rows of $G$. By Corollary~\ref{coro:generatorRank}, 
	$\R(\C)=\C\cup\langle \{2\vv_j *\vv_k\}_{1\leq j<k\leq \delta}\rangle$.	 
	By the same argument,  $\R(\C_Y)=\C_Y \,\cup \,\langle \{2v'_j 
	*v'_k\}_{1\leq j<k\leq \delta}\rangle$. Since $ 2\vv_j *\vv_k=(0\mid 2v'_j 
	*v'_k)$ for all $1\leq j<k\leq \delta$, the statement follows.
\end{proof}

The following theorem shows that if $\C$ is $\add$-additive
cyclic code, then $\R(\C)$ is also cyclic. As in the case of $\K(\C)$, this
theorem is a generalization of the case when $\C$ is a linear
cyclic code over $\Z_4$ \cite{Z4CyclicRK}.

\begin{theorem}\label{RankCyclic}
	Let $\C$ be a $\add$-additive cyclic code. Then $\R(\C) $ is $\add$-additive 
	cyclic code.
\end{theorem}
\begin{proof}
	Let $\textbf{x}\in\R(\C)$. By Corrollary~\ref{coro:generatorRank}, $\R(\C)$ is 
	generated by $\C$ and $\{2\textbf{v}*\textbf{w}\mid \textbf{v},\textbf{w}\in 
	\C\}$, then $\textbf{x}= \textbf{u} +
	2\textbf{v}*\textbf{w}$, for some $\textbf{u},\textbf{v},\textbf{w}\in\C$. As
	$\C$ is a $\add$-additive cyclic code then
	$\pi(\textbf{u}),\pi(\textbf{v}),\pi(\textbf{w})\in \C$ and
	$2\pi(\textbf{v})*\pi(\textbf{w})\in\R(\C)$. Thus, $\pi(\textbf{x})=
	\pi(\textbf{u})+2\pi(\textbf{v})*\pi(\textbf{w})\in\R(\C)$ and $\R(\C)$ is
	$\add$-additive cyclic code.
\end{proof}

The next proposition is straightforward from Theorems \ref{all-ranks}
and \ref{TypeDependingDeg}.

\begin{proposition}
	Let $\C$ be a $\add$-additive cyclic code. Then
	$$\alpha -\deg(b)+\deg(h) + 2 \deg(g)\leq \rank(\Phi(\C))\leq $$
	$$ \min (\alpha + \beta + \deg(g) - \deg(\gcd(b,\ell \tilde{g})), \; \alpha 
	-\deg(b)+\deg(h) + 2 \deg(g) + {\deg(g) \choose 2} ).$$
\end{proposition}

For a $\add$-additive cyclic code $\C$, define $\C_1 = \langle (b\mid 0)
\rangle$ and $\C_2 = \langle (\ell \mid fh+2f) \rangle$.
Since $\C=\C_1\cup \C_2$ and $\C_1\cap\C_2=\{\zero\}$, we have
that
\begin{equation}\label{rank1}
\rank(\Phi(\C)) = \rank(\Phi(\C_1)) + \rank(\Phi(\C_2)).
\end{equation}

If the code $\C$ is separable, then $\ell=0$ and 
$\rank(\Phi(\C_1))=\rank(\Phi(\C_X))$. Moreover, $\C_2=\langle
(0 \mid fh+2f) \rangle$, and therefore $\rank(\Phi(\C_2))=\rank(\Phi(\C_Y))$. We 
obtain the following result.

\begin{proposition}
	If $\C$ is a separable $\add$-additive cyclic code, then $\R(\C)=\R(\C_X)\times
	\R(\C_Y)$ and $\rank(\Phi(\C))=\kappa_1+\rank(\Phi(\C_Y))$.
\end{proposition}

Note that, if $\C$ is not separable, $\rank(\Phi(\C))$ is not necessarily equal 
to $\kappa_1+\rank(\Phi(\C_Y))$ as it
is shown in the following example.

\begin{example}
	Consider the $\add$-additive code generated by the following matrix.
	\begin{equation}
	\left(\begin{array}{ccc|ccc}
	1&0&0&0&0&0 \\
	0&1&0&0&0&0\\
	0&0&1&2&0&0\\
	0&0&0&1&1&0\\
	0&0&0&1&0&1\\
	\end{array}\right)
	\end{equation}
	In \cite[Example 2]{Z2Z4Image} it was proved that $\phi(\C_Y)$ is binary
	linear whereas $\Phi(\C)$ is not binary linear.
	In this example we have that $\kappa_1 =2$. Since $\Phi(\C_Y)$ is linear,
	$\rank(\Phi(\C_Y))=5.$  Nevertheless, $\Phi(\C)$ is not binary linear and $\rank 
	(\Phi(\C)) =
	8 > 5+2.$
\end{example}

In fact, $\rank(\Phi(\C))$ is always greater or equal to $\kappa_1 + 
\rank(\Phi(\C_Y))$.

\begin{proposition}
	Let $\C$ be a $\add$-additive cyclic code, then $\rank(\Phi(\C)) \geq
	\kappa_1 +\rank(\Phi(\C_Y))$.
\end{proposition}
\begin{proof}
	Let $\C$ be a $\add$-additive cyclic code. By (\ref{rank1}),
	$\rank(\Phi(\C))=\rank(\Phi(\C_1))+\rank(\Phi(\C_2))=\kappa_1+\rank(\Phi(\C_2))$
	. By
	Proposition~\ref{lemm:RankSet}, in order to determine the rank, we have to
	consider the set of vectors $2\vv*\vw$, for $\vv=(v\mid v'),\vw=(w \mid 
	w')\in\C$. Since for all
	$\vv\in\C$ if $\vw\in\C_1$ we obtain $2\vv*\vw=\zero$, we just have to consider
	$\vv,\vw\in\C_2$. We have that if $2 v' * w' \not \in \C_Y $ then
	$2 \vv*\vw \not \in \C.$ Therefore, $\rank(\Phi(\C_2))\geq 
	\rank(\Phi((\C_2)_Y))=\rank(\Phi(\C_Y))$
	and, therefore, $\rank(\Phi(\C_1))\geq \kappa_1+\rank(\Phi(\C_Y))$.
\end{proof}

Now we can determine the rank of a $\add$-additive code $\C$ as the rank of
$\C_X$ and the rank of a linear code over $\Z_4$, $\C'$. As in the case of the
kernel, when $\C$ is separable we have seen that $\C'=\C_Y$, but if $\C$ is not
separable such a code $\C'$ may not be cyclic over $\Z_4$.

\begin{theorem}
	Let $\C$ be a $\add$-additive cyclic code with generator matrix in the form of
	(\ref{gen}) and let $\C'$ be the subcode generated by the matrix in 
	(\ref{genSubcode}).  Then, $$\rank(\Phi(\C)) = \kappa_1 + \kappa_2 
	+\rank(\Phi(\C'_Y)).$$
	
\end{theorem}

\begin{proof}
	Let $\C$ be a $\add$-additive cyclic code with generator matrix $G$ in the form
	of (\ref{gen}). Let $\{\vu_i=(u_i\mid u'_i)\}_{i=1}^{\gamma}$ be the
	first $\gamma$ rows and  $\{\vv_j=(v_j\mid v'_j)\}_{j=1}^{\delta}$ the last
	$\delta$ rows of $G$.
	Define the codes
	$\bar{\C}=\langle\{\vu_i\}_{i=1}^{\kappa_1+\kappa_2}\rangle$ and
	$\C'=\langle\{\vu_i\}_{i=\kappa_1+\kappa_2+1}^{\gamma}
	,\{\vv_j\}_{j=1}^{\delta}\rangle$.
	
	By Corollary~\ref{coro:generatorRank} we have
	$$\R(\C)=\langle\{\vu_i\}_{i=1}^{\gamma},\{\vv_j\}_{j=1}^\delta, \{2\vv_j 
	*\vv_k\}_{1\leq j<k\leq \delta}\rangle,$$
	
	$$\R(\bar{\C})=\langle\{\vu_i\}_{i=1}^{\kappa_1+\kappa_2}\rangle,$$
	$$\R(\C')=\langle\{\vu_i\}_{i=\kappa_1+\kappa_2+1}^{\gamma}
	,\{\vv_j,2\vv_j\}_{j=1}^{\delta},\{2\vv_j *\vv_k\}_{1\leq
		j<k\leq \delta}\rangle.$$
	Note that $\R(\C)=\R(\bar{\C})\cup\R(\C')$. Moreover, 
	$\R(\bar{\C})\cap\R(\C')=\{\zero\}$ due to the fact that for all
	$1\leq j<k\leq \delta$, $2\vv_j *\vv_k=(0\mid 2v'_j
	*v'_k)\not\in\bar{\C}$. Therefore,
	$\rank(\Phi(\C))=\rank(\Phi(\bar{\C}
	))+\rank(\Phi(\C'))=\kappa_1+\kappa_2+\rank(\Phi(\C'))$.
	Finally, by Lemma~\ref{lemm:RankCY}, $\rank(\Phi(\C'))=\rank(\Phi(\C'_Y))$ and 
	the statement follows.
\end{proof}

\begin{theorem}\label{Theo:GeneratorsRank}
	Let ${\cal C}=\langle (b\mid 0), (\ell\mid fh +2f) \rangle$ be a $\add$-additive 
	cyclic code, where $fhg = x^\beta -1.$ Then,
	$\R(\C)=\langle (b_r\mid 0), ({\ell}_r\mid fh +2\frac{f}{r}) \rangle$, where $r$ 
	is a divisor of $f$ and $b_r$ divides $b$.
\end{theorem}
\begin{proof}
	By Theorem \ref{RankCyclic}, $\R(\C)$ is a $\add$-additive cyclic code,
	therefore $\R(\C)=\langle ({b}_r\mid 0), ({\ell}_r\mid {f}_r{h}_r
	+2{f}_r)\rangle$. Since $(b\mid 0)\in\R(\C)$, it is clear that ${b}_r$ divides 
	$b$. By \cite[Lemma 3]{Z2Z4RK}, $\C$
	and $\R(\C)$ have the same number of order four codewords and since $\C\subseteq
	\R(\C)$ we have that ${f}_r{h}_r=fh$. Then ${g}_r=g$. By Proposition
	\ref{f'_div_f}, we know that ${f}_r$ divides $f$ and hence there exists
	$r\in\Z_4[x]$ such that ${f}_r=\frac{f}{r}$ and $h_r=hr$. Therefore,
	$\R(\C)=\langle (b_r\mid
	0), ({\ell}_r\mid fh +2\frac{f}{r})\rangle$.
\end{proof}

Let $\C\langle (b\mid 0), (\ell\mid fh +2f) \rangle$ be a $\add$-additive cyclic 
code. In the following example, we will see that $\R(\C)=\langle (b_r\mid 0), 
({\ell}_r\mid f_rh_r +2f_r)\rangle$, where $b_r\not=b$, if we consider the 
generators of $\R(\C)$ in standard form. 

\begin{example}
	Let $x^7-1=(x-1)p_3q_3$ over $\Z_4$. Let $\C=\langle ((x-1)\mid 0), (1\mid 
	(x-1)+2) \rangle$, with $f=1, h=x-1$ and $g=p_3q_3$. If we compute $\R(\C)$ we 
	obtain that $\R(\C)=\langle (1\mid 0), (0\mid (x-1)+2) \rangle$.
\end{example}

%
$\add$-additive

%

As it is shown in Theorem~\ref{all-ranks}, there exists a $\add$-additive code
for any possible value of the rank. Nevertheless, the following example gives a
particular type $(\alpha, \beta; \gamma, \delta;\kappa)$ such that it is not
possible to construct a $\add$-additive cyclic code with a specific, and valid,
value of the rank.

\begin{example}\label{ex:Rank}
	Let $\C=\langle (b\mid 0), (\ell\mid fh +2f) \rangle$ be a $\add$-additive 
	cyclic code of type $(2,7; 2,3; \kappa)$. By Theorem~\ref{all-ranks}, 
	$\rank(\Phi(\C))\in\{8, 9, 10, 11\}$. We will see that there does not exist any 
	cyclic $\add$-additive code of type $(2,7; 2,3;\kappa)$, $\C$, with 
	$\rank(\Phi(\C))\in\{8, 9, 10 \}$.
	
	Let $x^7-1=(x-1)p_3q_3$ over $\Z_4$, with $p_3$ and $q_3$ as in 
	Example~\ref{ex:Kernel}. By Theorem~\ref{TypeDependingDeg}, $\deg(g)=3$ and 
	$\deg(b)=\deg(h)\leq 2$. Assume without loss of generality that $g=p_3$ and, 
	since $\deg(h)\leq 2$, we have that $q_3$ divides $f$.
	We have already proved, in Example~\ref{ex:Kernel}, that there does not exist a 
	$\add$-additive code of type $(2,7; 2,3; \kappa)$ with linear Gray image. Thus, 
	$\rank(\Phi(\C))\neq 8$.
	
	By Theorem~\ref{Theo:GeneratorsRank}, $\R(\C)=\langle (b\mid 0),
	({\ell}_r\mid f_rh_r +2{f}_{r}) \rangle$ where $r$ divides $f$ and $h_r=hr$. 
	Since $\rank(\Phi(\C))\in\{8, 9, 10, 11\}$, we have that $\deg(r)\leq 3$ as 
	$|\R(\C)|=4^32^{2+\deg(r)}\leq 2^{11}$.
	Since $\gcd(\tilde{q}_3, \tilde{p}_3\otimes\tilde{p}_3)\neq 1$ we have that
	$q_3$ must divide $r$. Therefore $\deg(r)\geq 3$, and by the previous argument, 
	we know that $r=q_3$. So, $\rank(\Phi(\C))\notin\{ 9, 10 \}$.
	
	Finally, we will give $\add$-additive cyclic codes of type $(2,7; 2,3;\kappa)$, for different values of $\kappa$, such that $\rank(\Phi(\C))=11$. 
	Recall that $\kappa\leq
	\min\{\alpha, \gamma\}=2$ and $\kappa= \alpha - \deg(\gcd(b,\ell \tilde{g}))$,
	then
	\begin{itemize}
		\item $\kappa=2$: $\C=\langle (x - 1\mid 0), (1\mid (x^3 + 2x^2 + x + 3)(x-1)
		+2(x^3 + 2x^2 + x + 3))\rangle$, or $\C=\langle (x - 1\mid 0), (1\mid (x^3 +
		3x^2 + 2x + 3)(x -1)
		+ 2(x^3 + 3x^2 + 2x + 3))\rangle$.
		\item $\kappa=1$: $\C=\langle (x - 1\mid 0), (0\mid (x^3 + 3x^2 + 2x + 3)(x -1)
		+ 2(x^3 + 3x^2 + 2x + 3))\rangle$, or $\C=\langle (x - 1\mid 0), (0\mid (x^3 +
		2x^2 + x + 3)(x-1)
		+2(x^3 + 2x^2 + x + 3))\rangle$.
		\item $\kappa=0$: As in Example~\ref{ex:Kernel}, there does not
		exist a $\add$-additive cyclic code of type $(2,7; 2,3; 0)$.
	\end{itemize}
\end{example}

\begin{proposition}\label{prop:RankQuaternary}
	Let $\C= \langle fh +2f \rangle$ be a cyclic code over $\Z_4$ of length 
	$n$, with $fhg=x^n-1$. Then, 
	$$
	\R(\C)= \big\langle fh +2\frac{f}{r} \big\rangle,
	$$
	where $r$ is the Hensel lift of $\gcd(\tilde{f},\tilde{g}\otimes 
	\tilde{g})$.
\end{proposition}

\begin{proof}
	From \cite{Z4CyclicRK}, we have that $\R(\C)=\langle fh +2\frac{f}{r} 
	\rangle$, where $r$ divides $f$. 
	
	Since $\R(\C)$ is the minimum cyclic code over $\Z_4$ containing $\C$ 
	whose image under the Gray map is linear, then $r$ is the polynomial of minimum 
	degree dividing $f$ satisfying that $\langle fh +2\frac{f}{r}\rangle$ has linear 
	image. This is equivalent, by \cite{wolfmann}, to the condition 
	$\gcd\big(\frac{\tilde{f}}{\tilde{r}},\tilde{g}\otimes \tilde{g}\big)=1$. 
	Therefore, the polynomial $r$ of minimum degree dividing $f$ satisfying this 
	condition is the Hensel lift of $\gcd(\tilde{f},\tilde{g}\otimes\tilde{g})$.
\end{proof}

\begin{proposition}
	Let ${\cal C}=\langle (b\mid 0), (\ell\mid fh +2f) \rangle$ be a 
	$\add$-additive
	cyclic code, where $fhg = x^\beta -1$, such that $\Phi(\C)$ is not 
	linear and $\Phi(\C_Y)$ is linear. Then, 
	$$
	\R(\C)=\langle (b_r\mid 0), ({\ell}_r\mid fh +2f) \rangle,
	$$	
	where $b_r=\gcd(b,\tilde{\mu}\tilde{g}{\ell})$, $\mu$ is as in 
	(\ref{eq:lambda-mu}), and $\ell_r={\ell}-\tilde{\mu}\tilde{g}{\ell} \pmod{b_r}$.
\end{proposition}

\begin{proof}
	Let ${\cal C}=\langle (b\mid 0), (\ell\mid fh +2f) \rangle$ be a 
	$\add$-additive cyclic code, $\C_Y=\langle fh +2f \rangle$. Let $G$ be a 
	generator matrix of $\C$ as in (\ref{eq:StandardForm}) and let $\{\vu_i=(u_i\mid 
	u'_i)\}_{i=1}^{\gamma}$ be the rows of order two and $\{\vv_j=(v_j\mid 
	v'_j)\}_{j=1}^{\delta}$ the rows of
	order four in $G$. By Corollary~\ref{coro:generatorRank}, 
	$\R(\C)=\C\cup\langle \{2\vv_j *\vv_k\}_{1\leq j<k\leq \delta}\rangle$. We have 
	that $\Phi(\C)$ is not linear, therefore there exist $i,k\in\{1,\dots,\delta\}$ 
	such that $2\vv_j *\vv_k=(\zero \mid 2v'_j *v'_k)\not \in \C$. Since 
	$\Phi(\C_Y)$ is linear, $2v'_j *v'_k\in\langle 2fh, 2f \rangle$. Moreover, 
	$\langle (0\mid 2fh) \rangle\in \C$ and therefore $\R(\C)=\C\cup\langle (0\mid 
	2f)\rangle$.
	Considering the generators polynomials of $\R(\C)$ and $\mu$ as in 
	(\ref{eq:lambda-mu}), we have 
	
	\begin{equation*}
	\def\arraystretch{1.5}
	\begin{array}{ccl}
	\R(\C) & = & \langle (b\mid 0), (\ell-\tilde{\mu}\tilde{g}\ell \mid 
	fh),(\tilde{\mu}\tilde{g}\ell \mid 2f),(0\mid 2f) \rangle\\
	& = & \langle (b\mid 0),(\tilde{\mu}\tilde{g}\ell \mid 0), 
	(\ell-\tilde{\mu}\tilde{g}\ell \mid fh+2f) \rangle\\
	& = &\langle (\gcd(b,\tilde{\mu}\tilde{g}{\ell})\mid 0), 
	(\ell-\tilde{\mu}\tilde{g}\ell \mid fh+2f) \rangle.\\
	\end{array}
	\end{equation*}
	
	Therefore, considering the polynomial generators of $\R(\C)$ in standard 
	form, we have that $b_r=	\gcd(b,\tilde{\mu}\tilde{g}{\ell})$ and 
	$\ell_r=\ell-\tilde{\mu}\tilde{g}\ell \pmod{b_r}$.
\end{proof}

\begin{theorem}\label{theo:rank-r}
	Let ${\cal C}=\langle (b\mid 0), (\ell\mid fh +2f) \rangle$ be a 
	$\add$-additive
	cyclic code, where $fhg = x^\beta -1$. Then,
	$$
	\R(\C)=\langle (b_r\mid 0), ({\ell}_r\mid fh +2\frac{f}{r}) \rangle,
	$$
	where  $r$ is the Hensel lift of $\gcd(\tilde{f},\tilde{g}\otimes 
	\tilde{g})$,
	$b_r=\gcd(b,\tilde{\mu}\tilde{g}{\ell})$, $\mu$ is as in (\ref{eq:lambda-mu}), 
	and $\ell_r={\ell}-\tilde{\mu}\tilde{g}{\ell}$.
	
\end{theorem}
\begin{proof}
	Let ${\cal C}=\langle (b\mid 0), (\ell\mid fh +2f) \rangle$ be a 
	$\add$-additive
	cyclic code. By Theorem~\ref{Theo:GeneratorsRank}, 
	$\R(\C)=\langle (b_r\mid 0), ({\ell}_r\mid fh +2\frac{f}{r}) \rangle$, where $r$ 
	is a divisor of $f$ and $b_r$ divides $b$.
	
	Consider the quaternary code $\C_Y=\langle fh +2f \rangle$. By 
	Lemma~\ref{lemm:RankCY}, we have that $(\R(\C))_Y=\R(\C_Y)$. Therefore, by 
	Proposition~\ref{prop:RankQuaternary}, $\R(\C_Y)=\big\langle fh +2\frac{f}{r} 
	\big\rangle,$
	where $r$ is the Hensel lift of $\gcd(\tilde{f},\tilde{g}\otimes 
	\tilde{g})$. Note that $\R(\C_Y)=\C_Y\cup \big\langle 2\frac{f}{r} \big\rangle$.
	
	From Corollary~\ref{coro:generatorRank}, $\C=\C\cup\langle 
	\{2\vv_j *\vv_k\}_{1\leq j<k\leq \delta}\rangle$, where $\{\vu_i=(u_i\mid 
	u_i')\}_{i=1}^{\gamma}$ are the rows of order four of the generator matrix of 
	$\C$ as in  (\ref{eq:StandardForm}). Note that, for all $\vu_j,\vu_k$, $1\leq 
	j\leq k \leq\delta$, $2\vu_j*\vu_k=(0\mid 2u'_j*u_k')$, where 
	$2u'_j*u_k'\in\R(\C_Y)=\C_Y\cup\langle 2\frac{f}{r} \big\rangle$ and, therefore, 
	$(0\mid 2u'_j*u_k')\in\C\cup\langle (0\mid 2\frac{f}{r}) \big\rangle$. Hence, we 
	have that $\R(\C)=\C\cup\langle (0\mid 2\frac{f}{r}) \big\rangle$. 		
	
	Therefore, for $\mu$ as in (\ref{eq:lambda-mu}),
	\begin{equation*}
	\def\arraystretch{1.5}
	\begin{array}{ccl}
	\R(\C) & = & \langle (b\mid 0), 
	(\ell-\tilde{\mu}\tilde{g}\ell \mid fh),(\tilde{\mu}\tilde{g}\ell \mid 
	2f),(0\mid 2\frac{f}{r}) \rangle\\
	& = & \langle (b\mid 0),(\tilde{\mu}\tilde{g}\ell \mid 
	0), (\ell-\tilde{\mu}\tilde{g}\ell \mid fh+2\frac{f}{r}) \rangle\\
	& = &\langle (\gcd(b,\tilde{\mu}\tilde{g}{\ell})\mid 0), 
	(\ell-\tilde{\mu}\tilde{g}\ell \mid fh+2\frac{f}{r}) \rangle.\\
	\end{array}
	\end{equation*}
	
	From the last equation, and considering the polynomial 
	generators of $\R(\C)$ in standard form, we have that $b_r= 
	\gcd(b,\tilde{\mu}\tilde{g}{\ell})$ and $\ell_r=\ell-\tilde{\mu}\tilde{g}\ell 
	\pmod{b_r}$.
\end{proof}

\begin{example}
	Let $\alpha=3$ and $\beta=7$. Consider, as in Example 3, $\C = \langle ( 
	x-1\mid 0), ( 0\mid x-1)\rangle$  where $f= x- 1$ and $h= 1$. As we have seen, 
	$\C$ does not have linear binary image. Then, by Theorem~\ref{theo:rank-r}, we 
	have that $\R(\C)= \langle (1\mid 0), (0\mid (x-1) + 2) \rangle$ where $r=f=x - 
	1$.
\end{example}

\begin{example}
	Let $\alpha=3$ and $\beta=15$. Consider $\C = \langle ( x-1\mid 0), ( 
	1\mid fh + 2f)\rangle$ where $f= x^4 + 2x^2 + 3x + 1$ and $h= (x-1)(x^4 + x^3 + 
	x^2 + x + 1)$. Then, by Theorem~\ref{theo:rank-r}, we have that $\R(\C)= \langle 
	(1\mid 0), (0\mid fh + 2\frac{f}{r}) \rangle$ where $r=f=x^4 + 2x^2 + 3x + 1$.
\end{example}

\section{Conclusions}
Given a $\add$-additive cyclic code $\C$, we have shown that the codes $\K(\C)$ and $\R(\C)$ are also $\add$-additive cyclic. We have computed the generator polynomials of these codes in terms of the generator polynomials of $\C$. Using these results, we have concluded that the dimensions of the binary images of $\K(\C)$ and $\R(\C)$, i.e. the dimension of the kernel and the rank of $\C$, cannot take all the possible values as for a general $\add$-additive code. In other words, if a $\add$-additive code is cyclic, then the set of possible values for the rank and the dimension of the kernel becomes more restrictive.

\end{document}